%% file: output.tex
\documentclass[sigconf]{acmart}
\newtheorem{theorem}{Theorem}
\newtheorem{lemma}{Lemma}

\usepackage{amsmath}
\usepackage{algorithmic}
\usepackage{comment}
\usepackage{times}
\usepackage[boxed,ruled,lined]{algorithm2e}

\usepackage{booktabs}
\usepackage{multirow}

\usepackage{comment}
\usepackage{balance}
\usepackage{bm}
\usepackage{subfigure}

\newcommand{\ie}{\emph{i.e., }}
\newcommand{\eg}{\emph{e.g., }}

\newcommand{\wrt}{\emph{w.r.t. }}

\input{math_commands.tex}

\usepackage{subcaption}

\usepackage{hyperref}
\usepackage{url}
\usepackage{graphicx}
\usepackage{tabularx}
\usepackage{array}

\title{A Taxation Perspective for Fair Re-ranking}

\copyrightyear{2024}
\acmYear{2024}
\setcopyright{acmlicensed}\acmConference[SIGIR '24]{Proceedings of the 47th
International ACM SIGIR Conference on Research and Development in
Information Retrieval}{July 14--18, 2024}{Washington, DC, USA}
\acmBooktitle{Proceedings of the 47th International ACM SIGIR Conference on
Research and Development in Information Retrieval (SIGIR '24), July 14--18,
2024, Washington, DC, USA}
\acmDOI{10.1145/3626772.3657766}
\acmISBN{979-8-4007-0431-4/24/07}

\author{Chen Xu}

\affiliation{%
  \institution{\mbox{Gaoling School of Artificial Intelligence}\\Renmin University of China}
    \country{xc\_chen@ruc.edu.cn}
}

\author{Xiaopeng Ye}
\affiliation{%
  \institution{Gaoling School of Artificial Intelligence\\Renmin University of China}
    \country{xpye@ruc.edu.cn}
}

\author{Wenjie Wang$^*$}
\affiliation{%
  \institution{NExT++ Research Center\\National University of Singapore}
   \country{wenjiewang96@gmail.com}
}

\author{Liang Pang}
\affiliation{%
  \institution{Institute of Computing Technology Chinese Academy of Sciences}
   \country{pangliang@ict.ac.cn}
}

\author{Jun Xu}
\authornote{Jun Xu and Wenjie Wang are corresponding authors. Work partially done at Engineering Research Center of Next-Generation Intelligent Search and Recommendation, Ministry of Education. }
\affiliation{%
    \institution{\mbox{Gaoling School of Artificial Intelligence}\\Renmin University of China}
    \country{junxu@ruc.edu.cn}
}

\author{Tat-Seng Chua}
\affiliation{%
  \institution{NExT++ Research Center\\National University of Singapore}
   \country{dcscts@nus.edu.sg}
}

\begin{document}

\renewcommand{\shortauthors}{Chen Xu, Xiaopeng Ye, Wenjie Wang, Liang Pang, Jun Xu, \& Tat-Seng Chua}

\begin{abstract}






Fair re-ranking aims to redistribute ranking slots among items more equitably to ensure responsibility and ethics. The exploration of redistribution problems has a long history in economics, offering valuable insights for conceptualizing fair re-ranking as a taxation process. Such a formulation provides us with a fresh perspective to re-examine fair re-ranking and inspire the development of new methods.
From a taxation perspective, we theoretically demonstrate that most previous fair re-ranking methods can be reformulated as an item-level tax policy.
Ideally, a good tax policy should be effective and conveniently controllable to adjust ranking resources.  However, both empirical and theoretical analyses indicate that the previous item-level tax policy cannot meet two ideal controllable requirements: (1) continuity, ensuring minor changes in tax rates result in small accuracy and fairness shifts; (2) controllability over accuracy loss, ensuring precise estimation of the accuracy loss under a specific tax rate.
To overcome these challenges, we introduce a new fair re-ranking method named Tax-rank, which levies taxes based on the difference in utility between two items.
Then, we efficiently optimize such an objective by utilizing the Sinkhorn algorithm in optimal transport. Upon a comprehensive analysis,
Our model Tax-rank offers a superior tax policy for fair re-ranking, theoretically demonstrating both continuity and controllability over accuracy loss.
Experimental results show that Tax-rank outperforms all state-of-the-art baselines in terms of effectiveness and efficiency on recommendation and advertising tasks.

\end{abstract}



\maketitle

\section{Introduction}

\begin{figure}[t]  
    \centering    
    \subfigure[Taxation process V.S. Fair re-ranking]
    {
        \includegraphics[width=0.95\linewidth]{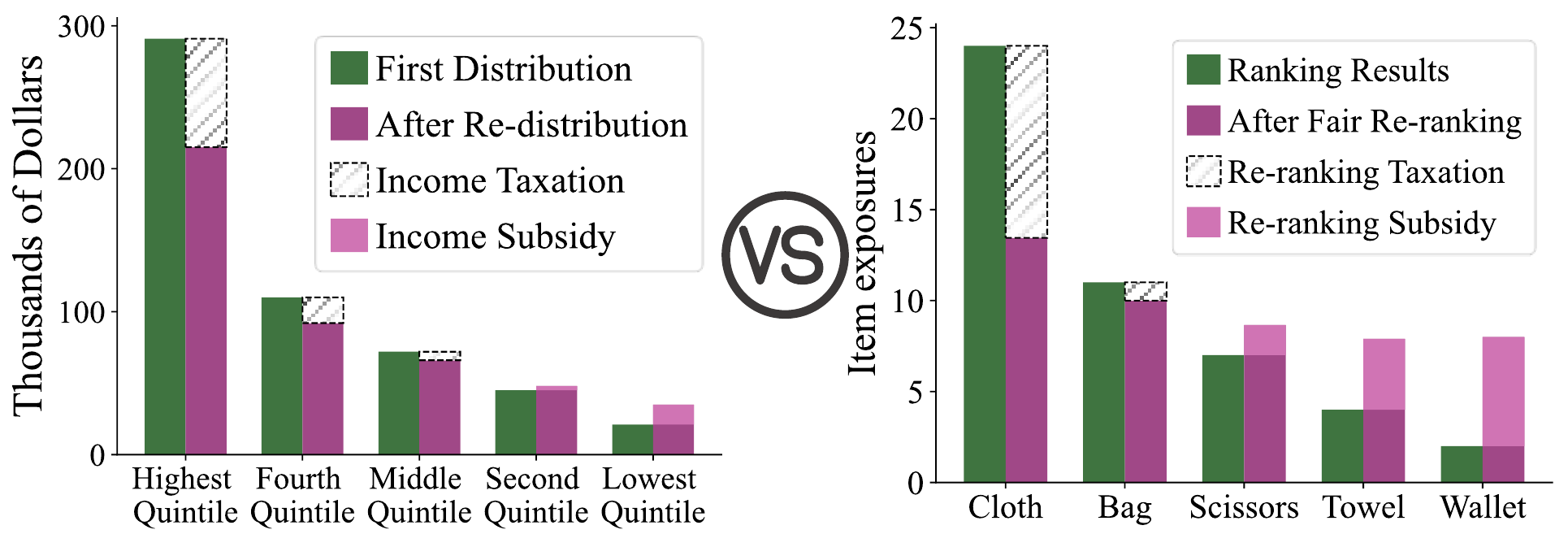}
    }
     \subfigure[Previous taxation process V.S. Our taxation process]
    {
        \includegraphics[width=0.95\linewidth]{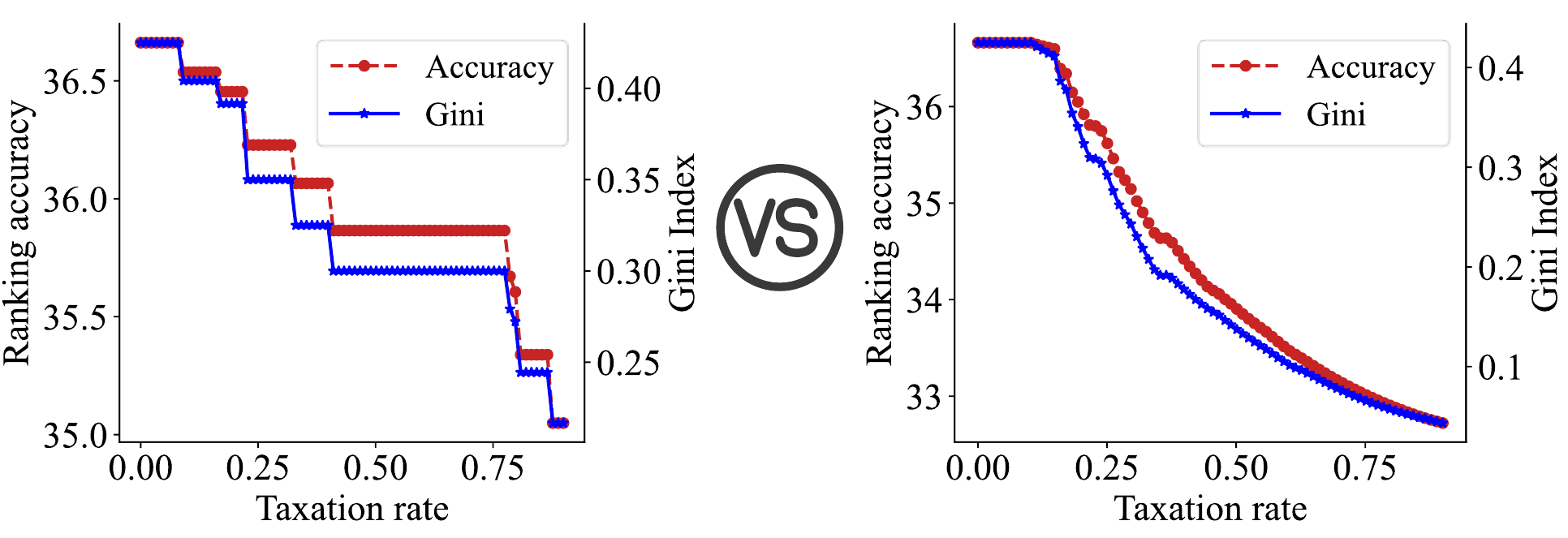}
    }
    \caption{Sub-figure (a) illustrates that fair re-ranking can be viewed as a taxation process from an economic perspective\protect\footnotemark.  Sub-figure (b) illustrates that our taxation process exhibits better continuity compared to previous ones.}
    \label{fig:intro}
\end{figure}

Ranking tasks such as recommendation and advertising is important in personalized information filtering 
for users~\cite{IRbook}. However, recent studies~\cite{fairrec, deldjoo2022survey} have revealed that prioritizing ranking accuracy will lead to significant unfair exposures of items, threatening the ethics and stability of ranking systems. Recently, many fair re-ranking methods~\cite{fairrec, deldjoo2022survey} have been developed to redistribute ranking slots among items more equitably by transforming the problem into a resource redistribution problem~\cite{xu2023p}.

The investigation of resource redistribution has a long history in economics~\cite{hanlon2010review, nerre2001concept}, offering economic insights for examining the task of fair re-ranking. 
For example, as shown in the left part of Figure~\ref{fig:intro} (a), federal taxes are structured to redistribute subsidies to the poor while aiming to minimize the sacrifice of social welfare.
Similarly, as shown in the right part of Figure~\ref{fig:intro} (a), the goal of fair re-ranking is to redistribute items, aiming for a more equitable distribution of ranking slots, without significantly compromising ranking accuracy. Since both taxing and fair re-ranking are essentially redistributing resources, we can examine fair re-ranking from the view of taxation. Intuitively, fair re-ranking can be viewed as taxing exposure from ``rich'' items (with more exposure) and redirects it to ``poor'' items (have less exposure) in Figure~\ref{fig:intro} (a).

Viewing fair re-ranking as a taxation process provides a fresh perspective to re-examine previous fair re-ranking methods and inspire new approaches.
From a taxation perspective, we can theoretically show that most previous fair re-ranking methods~\cite{wu2021tfrom, fairrec, xu2023p, Do22GeneralizedGini, cpfair} can be reformulated as an item-level tax policy, which imposes an additional tax on each item (see Section~\ref{sec:previous}).
In economics, a good tax policy should be effective and conveniently controllable to adjust ranking resources~\cite{hanlon2010review, nerre2001concept}. However, we observe that the item-level tax policy derived from previous fair re-ranking methods fails to meet these two ideal controllability criteria for taxation: (1) Continuity, implying that slight variations in tax rates lead to minor shifts in performances. As shown in the left part of Figure~\ref{fig:intro} (b),  When tax rates vary, previous methods exhibit numerous breakpoints in fairness and accuracy performance, indicating poor continuity. (2) Controllability over accuracy loss, ensuring an accurate estimation of accuracy loss caused by a specific tax rate. However, previous methods are uncertain about such accuracy loss (refer to Section~\ref{sec:re-examine} for a more detailed theoretical analysis). These two requirements are important for a tax policy since continuity ensures the stability of the tax rate, while controllability over accuracy loss ensures better management of the trade-off between fairness and accuracy.

\footnotetext{Left part is federal taxes of 2016 in the U.S., as reported by the Congressional Budget Office \url{https://www.cbo.gov}.}

To meet the two aforementioned essential criteria for the taxation process for fair re-ranking, in this paper, we introduce a novel fair ranking method, termed Tax-rank. Differing from the item-level tax policy, Tax-rank introduces a distinct fair re-ranking optimization objective, which levies taxes based on the difference in utility (\eg exposure) between two items. 
A detailed geometric explanation of such a taxation process is in Section~\ref{sec:geo}. Upon a comprehensive analysis, our taxation process presents two key advantages: (1) Tax-rank demonstrates greater effectiveness as it adheres to continuity \wrt tax rates, as depicted in the right part of Figure~\ref{fig:intro} (b). (2) Tax-rank offers enhanced controllability over accuracy loss, as we can provide an upper bound, showcasing the maximum accuracy loss across different tax rates. More detailed theoretical analyses are in Section~\ref{sec:two_requires}.

Meanwhile, we propose an effective algorithm to implement the taxation process of Tax-rank efficiently. We first present an easily solvable lower-bound function for the optimization objective of Tax-rank. However, the solution of the lower-bound function does not satisfy the ranking constraint, which requires fixed-sized items (i.e. Top-K ranking) for each user. To solve this issue, we utilize the Sinkhorn algorithm~\cite{swanson2020rationalizing} of optimal transport~\cite{pham2020unbalanced, peyre2019computational} to project non-compliant solutions onto solutions that satisfy the ranking constraint.


We summarize the major contributions of this paper as follows:

(1) We re-conceptualize the fair re-ranking task as a taxation process and re-examine the previous fair re-ranking methods from a novel taxation perspective.

(2) We introduce a novel fair re-ranking method named Tax-rank, incorporating a new optimization objective based on the taxation process and utilizing the Sinkhorn algorithm for efficient optimization.
Theoretical evidence indicates that Tax-rank exhibits superior continuity and systematic controllability.

(3) We conduct extensive experiments on two publicly available recommendation and advertising datasets, demonstrating that Tax-rank outperforms state-of-the-art baselines in terms of both fair ranking performance and efficiency.

\section{Related Work}
Recently, fair re-ranking tasks have become a compelling and pressing issue, driven by the need for a responsible and trustworthy ecosystem~\cite{lifairness, lipani2016fairness, deldjoo2022survey, patro2022fair}. Fairness concept in re-ranking varies for multi-stakeholders settings~\cite{abdollahpouri2020multistakeholder}, such as user-oriented fairness~\cite{abdollahpouri2019unfairness, li2021user} and item-oriented fairness~\cite{fairrec, xu2023p, Wang22makefair, singh2019policy}. The concept of user fairness suggests that ranking models should avoid delivering significantly disparate outcomes to users based on their sensitive attributes~\cite{matsumoto2016culture, tyler1995social,li2021user}. On the other hand, item fairness is more related to distributive justice~\cite{tyler2002procedural, lamont2017distributive}, which requires an equitable reallocation of ranking results for items within a healthy ecosystem. In this paper, our primary focus is on applying an economic perspective to reconsider item-fair re-ranking tasks.

Regarding item fair re-ranking, previous work can be divided into two types: one is regularized methods, which used a multi-task optimization approach with a linear combination of accuracy and fairness loss functions, incorporating a trade-off coefficient $\lambda$~\cite{xu2023p, Do22GeneralizedGini, cpfair}. Another approach employed constraint-based methods, formulating the task as a constrained optimization problem to ensure that fairness metrics do not exceed a specified threshold~~\cite{wu2021tfrom, fairrecplus, zafar2019fairness, fairrec}. For measuring fairness, both of the methods utilize different fairness metrics: ~\citet{ben2018game, fairrec, fairrecplus} proposed the Shapley value to optimize fairness; ~\citet{do2022optimizing} proposed to optimize Gini Index~\cite{do2022optimizing},  ~\citet{xu2023p, nips21welf} proposed to optimize max-min fairness function and ~\citet{jiang2021generalized} suggested to use the variance (distance) of different item utilities to measure the item fairness. Despite achieving significant performances, existing methods show poor continuity and systematic controllability under a taxation perspective for fair re-ranking.

In the economic field, the resource is usually allocated through first distribution and re-distribution process~\cite{lambert1992distribution}. In the process of redistribution, taxation is frequently employed as a mechanism to redistribute wealth and tackle income inequality~\cite{hanlon2010review, nerre2001concept}. There are usually two types of taxation methods: 
(1) flat tax rates, such as property tax~\cite{oates1969effects}, which are designed with different fixed rates based on the varying amounts of property;
(2) progressive tax, such as income taxes~\cite{poulson2008state} and payroll taxes~\cite{brittain1971incidence}, typically involve progressive tax rates that increase as a taxpayer's earnings increase. Process tax often is regarded as a more useful but complex taxation method.
Previous tax policies resemble a flat tax, whereas Tax-rank adopts a closer resemblance to a progressive tax. Finally, the Tax-rank policy is more closely related to $\alpha$-fair optimization in cooperative games~\cite{bertsimas2011price, bertsimas2012efficiency}. However, these approaches are not suitable for fair re-ranking tasks.

\section{Problem Formulation}\label{sec:formulation}



We first define some notations for the problem. For vector $\bm{x}\in\mathbb{R}^n$, let $\bm{x}_i$ denote the $i$-th element of the vector. For vector $\bm{x}\in\mathbb{R}^{n\times m}$, let $\bm{x}_{i,j}$ denote the element of $i$-th row and $j$-th column. 


In this section, we will first formulate the ranking task. Let 
$\mathcal{U}$ representing the set of users, and $\mathcal{I}$ representing the set of items.
When a user $u \in \mathcal{U}$ arrives, 
the ranking system will give the user $u$ a ranking list $L_K(u)\in\mathcal{I}^K$, which has fixed-size $K$ items.

Then, we define the item utility in the ranking task. Let $\bm{v}_i$ be the utility of item $i$. In the context of ranking, $\bm{v}_i$ is typically defined as the accumulated utilities of item $i$ across all ranking lists when $|\mathcal{U}|$ users arrive (\eg $\bm{v}_i$ could be total exposure or click numbers of item $i$ within a day). Formally, $\bm{v}_i$ is defined as $ \bm{v}_i = \sum_{u\in\mathcal{U}} w_{u,i}\bm{x}_{u,i}$, where $\bm{x}_{u,i}=1$ denotes item $i$ is in the ranking list of user $u$ (\ie $i\in L_K(u)$), otherwise, $\bm{x}_{u,i}=0$. If
$\bm{x}_{u,i}=1$, the item $i$ will receive an utility  
$w_{u,i}$. Typically, $w_{u,i}$ has two possible definitions: (1) item exposure~\cite{xu2023p, fairrec}, which implies that if an item is exposed to a user once, it will gain one unit of utility (\ie $w_{u,i}=1$). (2) item click~\cite{yang2019bid, liu2021neural}, which implies that if an item is clicked by a user, it will gain one unit of utility. Since we do not know whether a user will click an exposed item, we utilize the estimated probability of a user clicking on an item (\ie click-through-rate (CTR) value) as the value of $w_{u,i}\in [0,1]$.

Finally, the goal of fair re-ranking $f$ is to balance the various utilities of items more equitably. 
That is, on the one hand, $f$ aims to maximize the weighted sum of their utilities ($\sum_{i\in\mathcal{I}} \gamma_i\bm{v}_i$), where $\gamma_i$ is the item weight or bidding value~\cite{yang2019bid}. On the other hand, $f$ aims to achieve as even utilities ($\bm{v}_i \approx \bm{v}_j, \forall i,j\in\mathcal{I}$) as possible.

\section{Taxation Perspective for Fair Re-ranking}
In this section, we utilize the taxation perspective to re-think the fair re-ranking problem.


\subsection{Fair Re-ranking as Tax Policy}\label{sec:previous}

Firstly, the objectives of the previous fair re-ranking can be mainly divided into either regularization-based optimization~\cite{xu2023p, Do22GeneralizedGini, cpfair, fairrec} or constraint-based optimization~\cite{wu2021tfrom, fairrec, fairrecplus, zafar2019fairness}. Specifically, the objectives of regularization-based and constraint-based are formulated in Equation~(\ref{eq:lambda}) and Equation~(\ref{eq:constraint}), respectively:

  \begin{equation}\label{eq:lambda}
        \begin{aligned}
             W_1(\bm{x}) = \max_{\bm{x}\in\mathcal{X}} \quad& g(\bm{v}; \lambda_1) = \sum_{i\in\mathcal{I}} \gamma_i\sum_{u\in\mathcal{U}} w_{u,i}\bm{x}_{u,i} + \lambda r(\bm{v})
        \end{aligned},
    \end{equation}

    \begin{equation}\label{eq:constraint}
        \begin{aligned}
             W_2(\bm{x}) = \max_{\bm{x}\in \mathcal{X}} \quad& g(\bm{v}; \pi) = \sum_{i\in\mathcal{I}} \gamma_i\sum_{u\in\mathcal{U}} w_{u,i}\bm{x}_{u,i}\\
        \textrm{s.t.}\quad 
        &r(\bm{v}) \leq \pi
        \end{aligned},
    \end{equation}
where $\bm{v}_i = \sum_{u\in\mathcal{U}} w_{u,i}\bm{x}_{u,i}$, $\mathcal{X}=\{\bm{x}|\bm{x}_{u,i}=0/1, \sum_{i\in\mathcal{I}}\bm{x}_{u,i}=K\}$ is the feasible region of variable $\bm{x}$, $\lambda_1\in[0,\infty]$ is a trade-off coefficient, and $\pi$ is the pre-defined fairness threshold. The function $r(\bm{v})$ is the fairness function, which aims to measure the disparity of different items' utilities. Note that in previous literature, $r(\bm{v})$ has many forms: $\min_{i\in \mathcal{I}} \bm{v}_i$~\cite{xu2023p, fairrec, fairrecplus}, $-\sum_i |v_i-\bar{v}|$~\cite{cpfair, Jiakai} and any other forms~\cite{lifairness}. 

Then, in Theorem~\ref{theo:previous_tax}, we will illustrate that both of the objectives can be interpreted as a taxation process, which imposes an additional tax rate $\lambda$ and a taxation value $\bm{\mu}_i$ on each item $i$.

\begin{theorem}\label{theo:previous_tax}
Then optimal fair re-ranking result $\bm{x}(\lambda)$ with specific tax rate $\lambda$ can be achieved as:
    \begin{equation}\label{eq:optimial_value}
        \bm{x}^*(\lambda)= \argmax_{\bm{x}\in \mathcal{X}} \sum_{u\in\mathcal{U}}\sum_{i\in\mathcal{I}} s_{u,i}\bm{x}_{u,i},
    \end{equation}
    where  
    $s_{u,i} = \gamma_i w_{u,i} + \bm{\mu}_i,$ and
    \[\bm{\mu} = \argmin_{\bm{\mu}} \left[\sum_{u\in\mathcal{U}}\sum_{k=1}^K s_{u,[k]} +  \max_{\bm{v}} \left(\lambda r(\bm{v})-\bm{\mu}^{\top}\bm{v}\right) \right].\]
    
    In simpler terms, the ranking score is like the original score $\gamma_i w_{u,i}$ but with the addition of item-level taxation $\bm{\mu}_i$ with the tax rate of $\lambda\in[0,\infty]$. The tax rate $\lambda$ is calculated as
    \[
        \lambda =
        \begin{cases}
          \lambda_1, & \text{if } W_1(\bm{x}) \\
          \argmin_{\lambda_2} \left[\sum_{u\in\mathcal{U}}\sum_{k=1}^K c_{u,[k]} +  r(\bm{v})\lambda_2-\lambda_2\pi \right], & \text{if } W_2(\bm{x}) \\
        \end{cases},
    \]
    where $c_{u,i} = \gamma_i w_{u,i}$.

\end{theorem}

The proof of the Theorem~\ref{theo:previous_tax} can be seen in Appendix~\ref{app:previous_tax}. From Theorem~\ref{theo:previous_tax}, we can see that previous fair re-ranking methods can be regarded as an item-level tax policy with the taxation value as $\bm{\mu}_i$ for ranking score $s_{u,i}$. Meanwhile, we find that the tax policy is a well-studied knapsack problem~\cite{salkin1975knapsack}, where we can only choose the top K items with the highest scores $s_{u,i}$.

\subsection{Re-examine Fair Re-ranking}\label{sec:re-examine}
Through examination, we theoretically demonstrate that previous item-level tax policies of fair re-ranking lack continuity and controllability over accuracy loss.

\subsubsection{Non-continuity.} In Theorem~\ref{theo:non-continuty}, we will establish the non Lipschitz continuous~\cite{hager1979lipschitz} of the optimization objective for the existing taxation \wrt the tax rate $\lambda$:

\begin{theorem}\label{theo:non-continuty}
    When $r(v)$ is not continuous, $W_1(\bm{x},\lambda)$ and $W_2(\bm{x},\lambda)$ is non Lipschitz continuous wrt the tax rate $\lambda$. Formally, $\exists \varepsilon >0, \forall \delta>0$, when $|\lambda_0-\lambda|<\delta$,  $|W_*(\bm{x},\lambda)-W_*(\bm{x},\lambda_0)| \ge \varepsilon$, where $*=1,2$.
\end{theorem}

The detailed proof of Theorem~\ref{theo:non-continuty} can be seen in Appendix~\ref{app:non-continuty}. Intuitively, the binary variable $\bm{x}_{u,i}$ under the item-level policy and the non-differentiable property of some $r(\bm{v})$ lead to the non-continuity.
The Theorem~\ref{theo:non-continuty} tells us a slight change in tax rates $\lambda$ leads to significant shifts in fair re-ranking performances.

\subsubsection{Non-controllability over accuracy loss.} we will establish the definition of the price of taxation (POT)~\cite{bertsimas2011price} as the maximum ranking accuracy loss across various taxation levels, formulated as the ratio of ranking accuracy loss to the maximum accuracy:
\begin{equation}
    \textbf{POT}(\lambda) = \frac{\textbf{Acc}'(0)-\textbf{Acc}'(\lambda)}{\textbf{Acc}'(0)},
\end{equation}
where $\textbf{Acc}'(\lambda)$ denotes the ranking accuracy (here, accuracy refers to the overall user-item ranking scores that these items are exposed to by users) under existing tax rate $\lambda$, \ie
\[
    \textbf{Acc}'(\lambda) = \sum_{u\in\mathcal{U}}\sum_{i\in\mathcal{I}} \gamma_iw_{u,i}\bm{x}^*(\lambda)_{u,i},
\]
and $\bm{x}^*(\lambda)$ is the optimal result with specific tax rate $\lambda$. Upon examining Equation~(\ref{eq:optimial_value}), it becomes evident that the relationship between the ranking result $\bm{x}^*(\lambda)$ and the tax rate is non-linear and non-continuous (refer to Theorem~\ref{theo:non-continuty} and the non-linear of the fairness function $r(\bm{v})$~\cite{xu2023p, fairrec}). Simultaneously, it is noteworthy that Equation~(\ref{eq:optimial_value}) represents a well-studied integral knapsack problem~\cite{salkin1975knapsack}, and obtaining an upper bound for such a non-linear and non-continuous knapsack problem remains a huge challenge~\cite{cacchiani2022knapsack}. Therefore, we are uncertain about the extent to which a specific tax rate will incur accuracy losses.

\subsection{Insights from Taxation Perspective}
Through the lens of taxation, the fair re-ranking mechanism operates as a dynamic adjustment tool, adapting to various systematic objectives during both periods of economic prosperity and downturn. 
During periods of prosperity in the ranking system, there is typically a surge in user traffic, and providers are willing to offer a greater variety of items. However, the increased competition during such periods may result in unfairness, as it can lead to a long-tail effect and the monopolization of specific items~\cite{abdollahpouri2017controlling}.  During such periods, strategically increasing the tax rate can enhance fairness, fostering a healthier ecosystem~\cite{xu2023p}. 
On the other hand, during economic downturns in the ranking system, user traffic tends to decrease, and providers may exhibit reduced enthusiasm for participation. In such a period, appropriately reducing tax rates can incentivize provider competition, attracting more users to the system and ensuring economic vitality. Therefore, a good fair re-ranking method not only needs to be effective but also provide a conveniently controllable way to select an appropriate hyper-parameter (\eg tax rate) for regulating fairness degree.


\section{Tax-rank}
To overcome the previous tax policy issues, we propose an improved fair re-ranking method named Tax-rank.


\subsection{Optimization Objective}\label{sec:Tax-rank}
In this section, we will first introduce the proposed taxation optimization objective of Tax-rank. Drawing inspiration from the welfare function~\cite{bertsimas2012efficiency}, we levy taxes based on the difference in utility between two items: tax the higher utility of item $i$ and redirect to the item $j$ (\ie $v_i>v_j$) with the value as $\frac{\gamma_j}{\gamma_i}(\frac{\bm{v}_i}{\bm{v}_j})^t$. Then the overall optimization objective will be

\begin{equation}\label{eq:Tax-rank}
    \begin{aligned}
         \bm{x}^*(t) = \argmax_{\bm{x}\in\mathcal{X}_s} &f(\bm{x}; t) =\begin{cases} 
      \sum_i \gamma_i \bm{v}_i^{1-t}/(1-t) & \text{if } t \ge 0, t \neq 1 \\
      \sum_i \gamma_i\log(\bm{v}_i)& \text{if } t=1 \\
    \end{cases}\\
    \textrm{s.t.}\quad  
    &\bm{v}_i = \sum_{u\in\mathcal{U}} w_{u,i}\bm{x}_{u,i}, \quad \forall i\in \mathcal{I}\\
    \end{aligned},
\end{equation}
where $\mathcal{X}_s = \{\bm{x}|\mathbf{x}_{u,i}\in [0,1],\sum_{i\in\mathcal{I}}\mathbf{x}_{u,i} = K, \forall u\in \mathcal{U}\}$ is the feasible region of variable $\bm{x}$. Intuitively, to keep the social welfare function 
$f(\bm{x}; t)$ unchanged, for every additional unit of utility gained by the ``poor'' item $j$, the ``rich'' item $i$ has to pay the tax value of $\frac{\gamma_j}{\gamma_i}(\frac{\bm{v}_i}{\bm{v}_j})^t$ with their weight $\gamma_i$ and $\gamma_j$.

Note that, Tax-rank introduces a relaxation of the binary solution for $\bm{x}$, transitioning it into a continuous version. Here, $\bm{x}_t$ can be interpreted as the recommendation probability assigned to each item, allowing us to leverage a multi-nominal sample technique to generate ranking lists. $t\in [0,\infty)$ represents the tax rate, where a higher value results in a more fair results. 

When $t=0$, the optimization objective reduces to the accuracy-first objective; When tax rate $t=1$, the objective can be reduced to Nash solution~\cite{nash1950bargaining}, where the tax rate will make the item utilities proportional to its weight, \ie $\bm{v}_i:\bm{v}_j = \gamma_i : \gamma_j, \forall i,j\in\mathcal{I}$. When $t$ approaches $\infty$, the optimization objective reduces to max-min form $\min_{i\in\mathcal{I}} \bm{v}_i$, leading to even utility for every item.


\begin{figure}[t]  
    \centering    
    \includegraphics[width=0.9\linewidth]{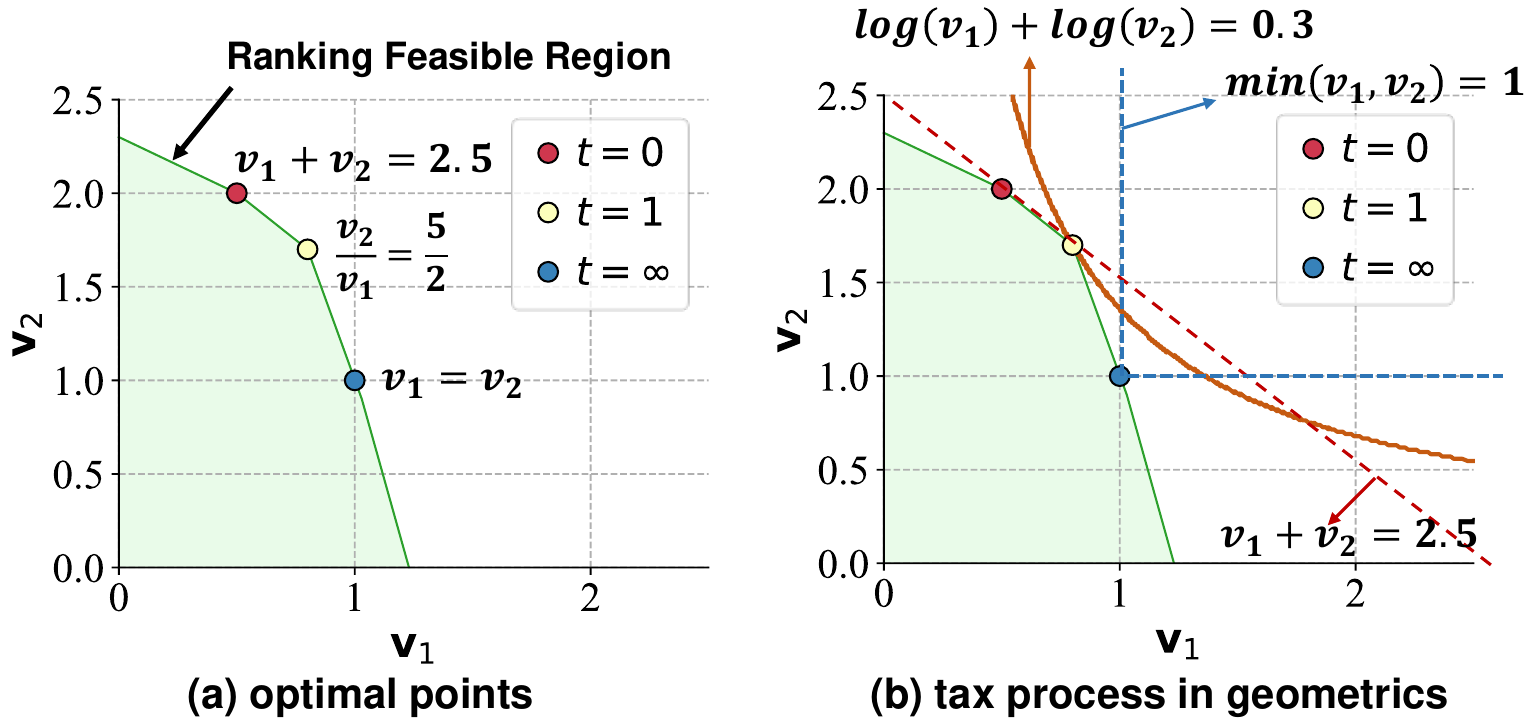}
    \caption{Geometric explanation for our taxation process, which imposes taxes based on the disparity in utility between two items. }
    \label{fig:example_taxation}
\end{figure}



\subsection{Optimization Objective Analysis}\label{sec:Tax-rank}
In this section, we will give a formal analysis of our optimization objective in a geometric view and its controllable requirements.

\subsubsection{Geometric explanation for our taxation process}\label{sec:geo}

To gain a better understanding of Tax-rank, we visualize the geometric to show our taxation process for taxing the difference in utility between two items.
In Figure~\ref{fig:example_taxation}, the ranking system only has two items and ranking size $K=1$.  The x-axis $\bm{v}_1$ and y-axis $\bm{v}_2$ represent the expected click numbers of items 1 and 2 (see Section~\ref{sec:formulation}). We set their weight $\gamma_1=\gamma_2=1$.
The green area describes the feasible region of all the possible values of $\bm{v}_1$ and $\bm{v}_2$ under ranking constraint.
The circle points are optimal solutions under different tax rates $t$. The line in Figure~\ref{fig:example_taxation} (b) represents the tangent between the optimization objective and the feasible region.

In Figure~\ref{fig:example_taxation} (a), 
we observe that when $t=0$ (red points), prioritizing overall utility $\bm{v}_1+\bm{v}_2=2.5$; at $t=1$ (yellow points), the tax rate will make the item utilities proportional to its weight, \ie $\bm{v}_1:\bm{v}_2=2:5$. At $t=\infty$ (blue points), 
the two items will share the same utility, \ie $\bm{v}_1=\bm{v}_2$. In Figure~\ref{fig:example_taxation} (b), we give a geometric explanation to show the optimal points change \wrt tax rate. From Figure~\ref{fig:example_taxation}, the optimal points must stay in the boundary lines (\ie Pareto frontier) of the ranking feasible region (see detailed proof in ~\cite{bertsimas2011price}). Therefore, the optimal points should be the tangent points between the contours of the optimization objective and the Pareto front. Let's consider the taxation process in two-dimensional space.

For example, in Figure~\ref{fig:example_taxation} (b), when $t=0$, the red point is the tangent point between the red line $\bm{v}_1+\bm{v}_2 = 2.5$ and the Pareto front. When $t = 1$, a tax of $\frac{\bm{v}_2}{\bm{v}_1}$ is applied to from $\bm{v}_2$ to $\bm{v}_1$ along the red line if $\bm{v}_2>\bm{v}_1$, which means the slope $\frac{\partial \bm{v}_2}{\partial \bm{v}_1}$ becomes $\frac{\bm{v}_1}{\bm{v}_2}$. When $\bm{v}_1>\bm{v}_2$, similar operation also make the slope $\frac{\partial \bm{v}_2}{\partial \bm{v}_1}$ becomes $\frac{\bm{v}_1}{\bm{v}_2}$.
Then the taxation process leads to the optimization function of $\log \bm{v}_1 + \log \bm{v}_2$ (orange line). Similarly, $t=\infty$, a tax of $(\frac{v_2}{v_1})^\infty$ is applied at each point on the red line, resulting in a slope $\frac{\partial \bm{v}_2}{\partial \bm{v}_1}$ of each point become $\infty$ and $0$ for each point when $\bm{v}_2>\bm{v}_1$ and $\bm{v}_1>\bm{v}_2$, respectively. Such taxation process leads to the objective becomes $\min (\bm{v}_1, \bm{v}_2)$ (blue line). Without loss of generality, tax the value as $\frac{\gamma_j}{\gamma_i}(\frac{\bm{v}_i}{\bm{v}_j})^t$ for every $\bm{v}_i>\bm{v}_j$ will transform the objective into the Equation~(\ref{eq:Tax-rank}).

In summary, the tax process of imposing taxes based on the disparity in utility between two items will determine the optimization objective of Equation~(\ref{eq:Tax-rank}) by influencing the slope of each point in the geometric perspective. 



\subsubsection{Controllable requirements}\label{sec:two_requires}
we demonstrate the optimization objective meets two ideal
controllable requirements: continuous and controllability over the accuracy loss.

\textbf{Continuity.} Firstly, the Tax-rank policy $f(\bm{x};t)$ adheres to the Lipschitz continuity property~\cite{hager1979lipschitz} of the optimization objective associated with the existing taxation concerning the tax rate $t$. Because we can easily observe that $f(\bm{x};t)$ takes an exponential form \wrt $t$, and their first-order derivatives are also continuous~\cite{lan2010axiomatic}. Meanwhile, transforming the binary solution of previous methods to a continuous solution also ensures feasible region is continuous. 

\textbf{Controllability over accuracy loss.} For the bound the maximum ranking accuracy loss across different taxation degrees under varying tax rates $t$, we will have the following Theorem:

\begin{theorem}\label{theo:POF}
       The price of taxation (POT) of Tax-rank is bounded:
    \begin{equation}\label{eq:POF}
        \text{POT} = \frac{\textbf{Acc}(0)-\textbf{Acc}(t)}{\textbf{Acc}(0)} \leq 1 - O(|\mathcal{U}|^{-\frac{t}{1+t}}),
    \end{equation}
    where $\textbf{Acc}(t)$ denotes the accuracy under Tax-rank tax policy with tax rate $t$, \ie
$
    \textbf{Acc}(t) = \sum_{u\in\mathcal{U}}\sum_{i\in\mathcal{I}} \gamma_iw_{u,i}\bm{x}^*(t)_{u,i},
$
$\bm{x}^*(t)$ is the  optimal fair re-ranking result with specific tax rate $t$ in Equation~(\ref{eq:Tax-rank}).
\end{theorem}

Detailed proof can be seen in Appendix~\ref{app:POF}. As we can see from Theorem~\ref{theo:POF}: when increasing the tax rate $t$ in a ranking system, there is a bound on the rate $1-O(|\mathcal{U}|^{-\frac{t}{1+t}})$ at which ranking utilities will decrease. The bound offers fair re-ranking systems a controllable way to determine the appropriate tax rate $t$.


\subsection{Algorithm}
The overall algorithm workflow can be seen in Algorithm~\ref{alg:alpha}. We observe that directly optimizing the Equation~(\ref{eq:Tax-rank}) is NP-hard since it is a non-linear, large-scale, and integral programming~\cite{bertsekas1997nonlinear}. Therefore, firstly, we construct an easy-solved standard programming (lines 1-2), which is the lower bound function of Equation~(\ref{eq:Tax-rank}). Then we apply the OT projection to obtain the final ranking result  (lines 3-9).

\subsubsection{Lower Bound Function Construction}\label{sec:bound}
\begin{lemma}\label{theo:upperbound}
    There exists $\tau>0$, s.t. we have
    \begin{equation}\label{eq:upperbound}
       \begin{aligned}
            f(\bm{x}; t)& \ge \widetilde{f}(\bm{x}; t)= \max_{\bm{e}} \sum_{i} \gamma_i \eta_i g(\bm{e};t)\\
             \quad\textrm{s.t.} &\quad \sum_{i\in\mathcal{I}} \bm{e}_i = K, \quad 0 \leq \bm{e}_i \leq 1, \quad \eta_i = \tau \sum_{u\in\mathcal{U}} w_{u,i}, \forall i\in\mathcal{I}\\
             &\quad g(\bm{e};t)=\begin{cases} 
      \sum_i \bm{e}_i^{1-t}/(1-t) & \text{if } t \ge 0, t \neq 1 \\
      \sum_i \log (\bm{e}_i)& \text{if } t=1 \\
   \end{cases},
       \end{aligned}
    \end{equation}
\end{lemma}

This lemma can be easily achieved when $
\sum_{u\in\mathcal{U}}\sum_{i\in\mathcal{I}} w_{u,i}\bm{x}_{u,i}\ge \sum_{i\in\mathcal{I}}\eta_i\sum_{u\in\mathcal{U}} \bm{x}_{u,i}.$ Let $\bm{e}^*$ be the optimal value of the variable $\bm{e}$, which represents the accumulated exposure of items over all users. Then we will apply the Sinkhorn algorithm~\cite{pham2020unbalanced} to project the averaged exposure $\bm{e}^*$ to recommendation decision variable $\bm{x}\in \mathcal{X}_s$ discussed in Section~\ref{sec:formulation}. 

\subsubsection{Optimal Transport Projection}
\label{sec:OT}

We obtain the final ranking result by utilizing the following sample process, where $\widetilde{\bm{x}}$ (i.e. ranking score distribution) is derived from the OT projection process.
\begin{equation}\label{eq:ranking_result}
    L_K(u) = \text{Sample}_{S \subset \{1, 2, \ldots, |\mathcal{I}|\}, |S| = K} \widetilde{\bm{x}}_{u}, \quad \forall u\in\mathcal{U},
\end{equation}
where it implies sampling $K$ non-repeated items to user $u$ according to the recommended probability $\widetilde{\bm{x}}_{u,i}\in[0,1]$ of item $i$.


We construct a matrix $\bm{C} = \mathbb{R}^{|\mathcal{U}|\times |\mathcal{I}|}$, where the element $\bm{C}_{u,i} = \gamma_i w_{u,i}$.
An OT problem can be formulated as:
\begin{equation}\label{eq:sinkhorn}
    \begin{aligned}
         \widetilde{\bm{x}} &= \argmin_{\bm{x}\ge 0} \langle \bm{x}, -\bm{C}\rangle + \lambda_{ot} H(\bm{x}) \quad  \textrm{s.t.} \quad \bm{x}\bm{1} = K\bm{1}, \quad \bm{1}^{\top}\bm{x} = \bm{e}^*
    \end{aligned},
\end{equation}
where $\bm{1}$ denotes a vector of ones, $\bm{e}^*$ denotes the optimal value of Equation~(\ref{eq:upperbound}) and $\lambda_{ot}$ is the coefficient of entropy regularizer. 
$\langle \bm{x}, -\bm{C}\rangle$ results transport plan lies on the Pareto frontier and
$
H(\bm{x}) = \sum_{u\in\mathcal{U}}\sum_{i\in\mathcal{I}} \bm{x}_{u,i}\log(\bm{x}_{u,i}),
$ which forces the variable $\bm{x}_{u,i}$ into the feasible region $[0,1]$. The constraint condition ensures that the ranking satisfies the limitation that each user can only be ranked among the top $K$ items, and it also guarantees that the exposure of each item aligns optimally with the predefined exposure vector $\bm{e}^*$.

This problem can be efficiently solved by the Sinkhorn algorithm~\cite{swanson2020rationalizing}, where the solution of the
form $\widetilde{\bm{x}} = \text{diag}(\bm{m})\bm{B}\text{diag}(\bm{n})$,
where diag$(\cdot)$ denote the generating diagonal matrix from vector,$ \bm{B} = e^{\frac{-C}{\lambda}}$, and $\bm{m}\in \mathbb{R}^{|\mathcal{U}|}$, $\bm{n}\in \mathbb{R}^{|\mathcal{I}|}$, which iteratively computes
$
    \bm{m} \leftarrow K\bm{1} \oslash \bm{B}\bm{n},\quad \bm{n} \leftarrow  \bm{e}^* \oslash \bm{B}\bm{m},
$
where $\oslash$ denotes element-wise division.

\begin{algorithm}[t]
    \caption{Learning algorithm of Tax-rank}
	\label{alg:alpha}
	\begin{algorithmic}[1]
	\REQUIRE User set $\mathcal{U}$, item set $\mathcal{I}$, ranking size $K$, tax rate $t$, OT coefficient $\lambda_{ot}$, item weight $\gamma_i, \forall i\in\mathcal{I}$, user-item ranking score $w_{u,i}, \forall u\in\mathcal{U},\forall i\in\mathcal{I}$ 
   
	\ENSURE The ranking result $L_K(u), \forall u\in \mathcal{U}$
	    \STATE Get the optimal averaged exposure $\bm{e}^*$ from Equation~(\ref{eq:upperbound}).
        \STATE Initialize $\bm{m} = K \bm{1}$,  $\bm{n} = \bm{e}^*$, $\bm{C}_{u,i} = \gamma_i w_{u,i}, \forall u\in\mathcal{U},\forall i\in\mathcal{I}, \bm{B} = e^{\frac{-C}{\lambda_{ot}}}$
        \FOR{$t=1,\cdots,T$}
            \STATE $\bm{m} = K\bm{1} \oslash \bm{B}\bm{n}$, $\bm{n} =  \bm{e}^* \oslash \bm{B}\bm{m}$
        \ENDFOR
        \STATE $\widetilde{\bm{x}} = \text{diag}(\bm{m})\bm{B}\text{diag}(\bm{n})$
        \STATE $L_K(u) = \text{Sample}_{S \subset \{1, 2, \ldots, |\mathcal{I}|\}, |S| = K} \widetilde{\bm{x}}_{u}, \quad \forall u\in\mathcal{U}$
	\end{algorithmic}
\end{algorithm}

\begin{figure*}
    \centering
    \includegraphics[width=0.95\linewidth]{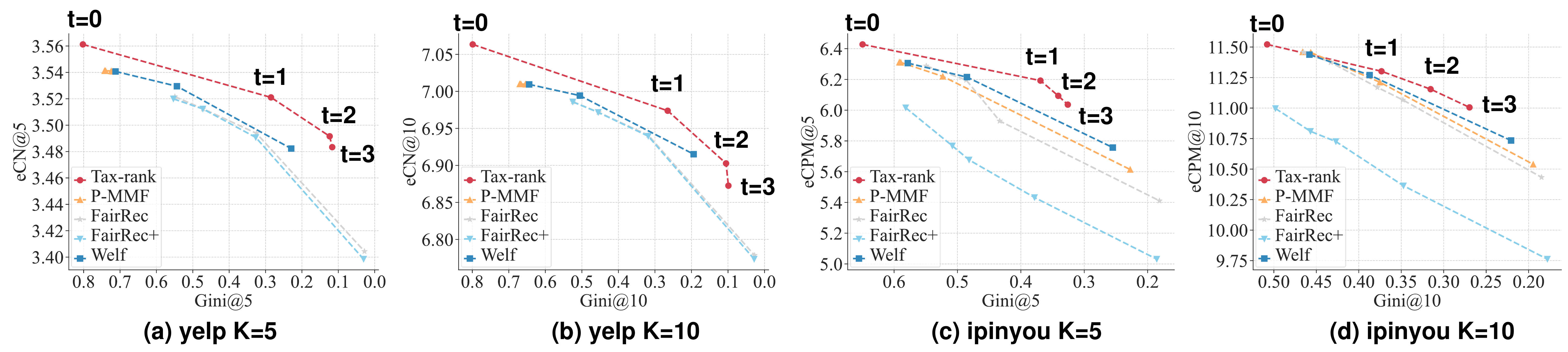}
    \caption{Pareto frontier with different top-K ranking under CTR-based settings (\ie $w_{u,i}$ is the CTR value of user-item pair).}
    \label{fig:Pareto}

\end{figure*}

\begin{figure*}
    \centering
    \includegraphics[width=0.95\linewidth]{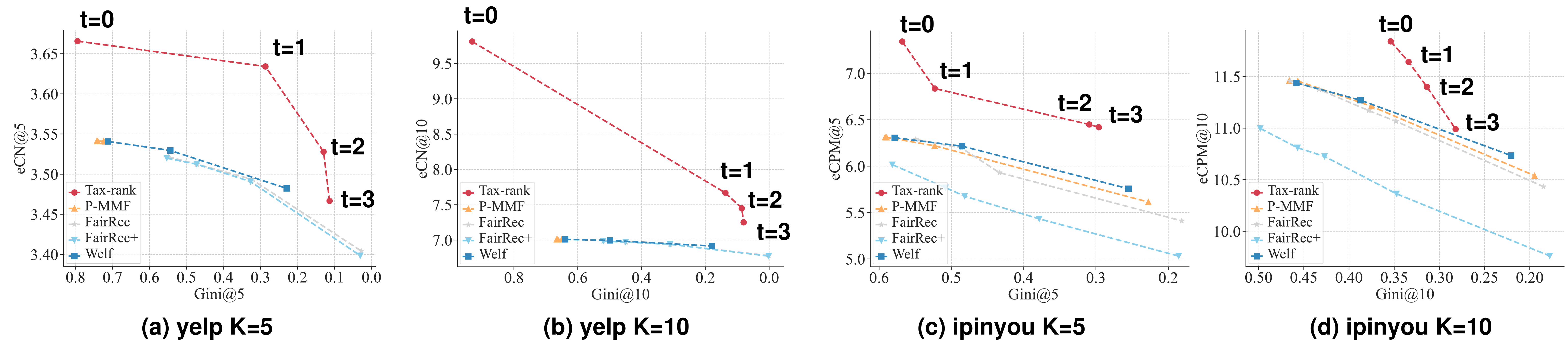}
    \caption{Pareto frontier with different top-K ranking under exposure-based settings (\ie $w_{u,i}=1$).}
    \label{fig:pareto_expo_based}
\end{figure*}


\section{Experiment}
We evaluate Tax-rank using two publicly available ranking datasets\footnote{The source codes have been shared in \url{https://github.com/XuChen0427/Tax-rank}}.

\subsection{Experimental settings}

\textbf{Dataset.} The experiments were based on two large-scale, publicly available ranking applications, including:
\textbf{Yelp}\footnote{\url{https://www.yelp.com/dataset}}: a large-scale businesses recommendation dataset. It has 154543 samples, which contain 17034 users, 11821 items.
\textbf{Ipinyou}~\cite{liao2014ipinyou}\footnote{\url{http://contest.ipinyou.com/}}: a large-scale advertising dataset. We only use the clicked data, which contains 18588 samples, which contains 18565 users, and 149 advertisements.
During the pre-processing step, users and items that had interactions with fewer than 5 items or users were excluded from the entire dataset to mitigate the issue of extreme sparsity. 

\textbf{Evaluation.} As for the evaluation metrics, the performances of the models were evaluated from two aspects: ranking accuracy, and fairness degree. As for the ranking accuracy, following the practices in~\cite{wu2021tfrom, xu2023ltp, yang2019bid}, we utilize excepted Click/Exposure Number (eCN) for recommendation and expected Cost Per Mile (eCPM) for advertising under top-$K$ ranking.:
\begin{equation}
    \text{eCN@K} = \frac{1}{|\mathcal{U}|}\sum_{i\in\mathcal{I}} \bm{v}_i,\quad
    \text{eCPM@K} = \frac{1}{|\mathcal{U}|}\sum_{i\in\mathcal{I}} \text{bid}_i\bm{v}_i,
\end{equation}
where, $bid_i$ denotes the bidding price of an advertisement, while $\bm{v}_i$ is calculated using Equation~\ref{eq:Tax-rank}, which is dependent on the $K$.

As for the fairness degree, we utilize the Gini Index~\cite{do2022optimizing, nips21welf}, which is the most common measure of item utility inequality under top-$K$ ranking. Formally, it is defined as:
\begin{equation}
    \text{Gini@K} = \frac{\sum_i \sum_j |\gamma_i\bm{v}_i-\gamma_j\bm{v}_j|}{2|\mathcal{I}|\sum_i \gamma_i\bm{v}_i},
\end{equation}
where it ranges from 0 to 1, with 0 representing perfect equality (every item has the same utility), and 1 representing perfect inequality (one item has all the utility, while every item else has none).

\textbf{Baselines.}
The following representative item fairness models were chosen as the baselines: \textbf{FairRec}~\cite{fairrec} and \textbf{FairRec+} \cite{fairrecplus} proposed to ensure Max-Min Share ($t$-MMS) of exposure for the items. \textbf{P-MMF}~\cite{xu2023p} utilized the mirror descent method to improve the worst-off item's utility. Moreover, we also compare \textbf{Welf}~\cite{nips21welf}, which also used the Frank-Wolfe algorithm to optimize a similar exponential form of optimization objective for two-sided fairness.

\textbf{Implementation details.} Following~\cite{xu2023p}, CTR-based settings, we use BPR~\cite{BPR} model to compute the CTR value $w_{u,i}\in [0,1]$ of each user-item pair $(u,i)$. For the item weight $\gamma_i$, 
a value of 1 is assigned for recommendation applications, while for advertising applications, $\gamma_i = \log(\text{bid}_i)$.

As for the hyper-parameters in all models, the trade-off coefficient ($\lambda_{ot}$) for OT  was tuned among $[0.1,2]$, and we implemented Tax-rank utilizing cvxpy~\cite{cvxpy} to conduct optimization for Equation~(\ref{eq:upperbound}).

\subsection{Experimental Results}
Figure~\ref{fig:Pareto} and Figure~\ref{fig:pareto_expo_based} shows the Pareto frontiers~\cite{xu2023p} of Gini Index (abbreviated as Gini.) and eCN/eCPM on two application datasets with different ranking size $K$. Figure~\ref{fig:Pareto} is generated using CTR-based settings, as described in Section~\ref{sec:formulation}, where the item utilities are defined as the expected number of clicks within a specified time period. On the other hand, Figure~\ref{fig:pareto_expo_based} is generated using exposure settings (also see Section~\ref{sec:formulation}), where the item utilities are defined as the expected number of exposures.
The Pareto frontiers~\cite{lotov2008visualizing} are constructed by systematically adjusting various parameters of the models and then selecting the points with the best performance in terms of both Gini@K and eCN@K/eCPM@K, resulting in an optimized trade-off between item fairness and total utilities.

Firstly, it is evident that a trade-off exists between ranking accuracy metrics (eCN@K or eCPM@K) and item fairness metric (Gini@K) \wrt the tax rate $t$. As the tax rate approaches $0$, FairTax tends to prioritize ranking accuracy, and as the tax rate increases, FairTax tends to emphasize item fairness.

Moreover, compared with the baseline methods, it becomes evident that the proposed Tax-rank method consistently outperforms the baseline methods under both CTR-based and exposure-based settings (as indicated by the  Tax-rank curves occupying the upper right corner). This Pareto dominance signifies that, for a given eCN@K or eCPM@K level,  Tax-rank achieves superior Gini@K values, and for a given Gini@K level, it attains better eCN@K or eCPM@K performance. These results highlight that  Tax-rank significantly outperforms the baseline methods.

\begin{figure*}[t]  
    \centering    
    \includegraphics[width=0.95\linewidth]{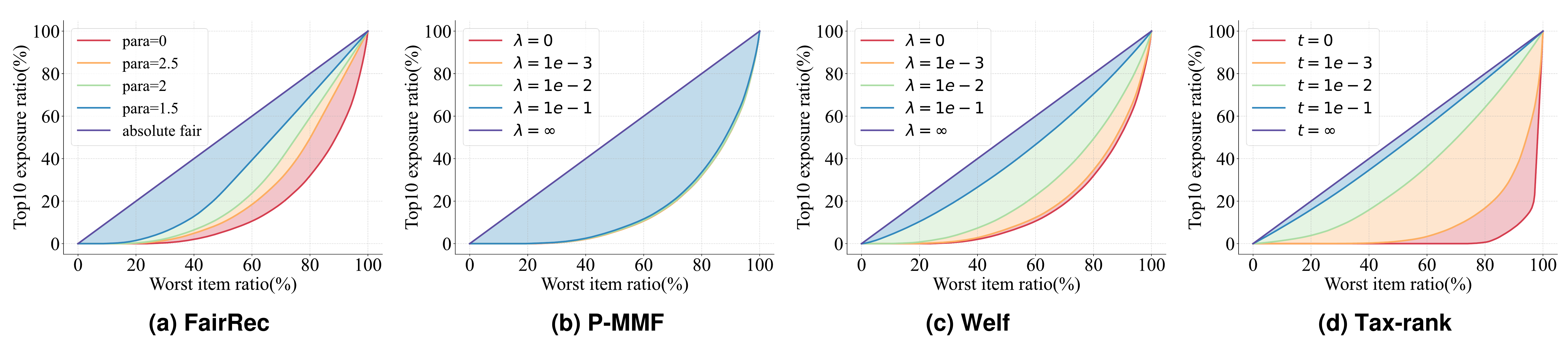}
    \caption{Lorenz Curve~\cite{Lorenzcurve} of three best-performing baselines FairRec, P-MMF and Welf and our model Tax-rank. The distinct curves in each figure are plotted by adjusting various tax rates $t$ or different parameters. }
    \label{fig:Lorenz_Curve}  
\end{figure*}

\begin{figure}[t]  
    \centering    
    \includegraphics[width=\linewidth]{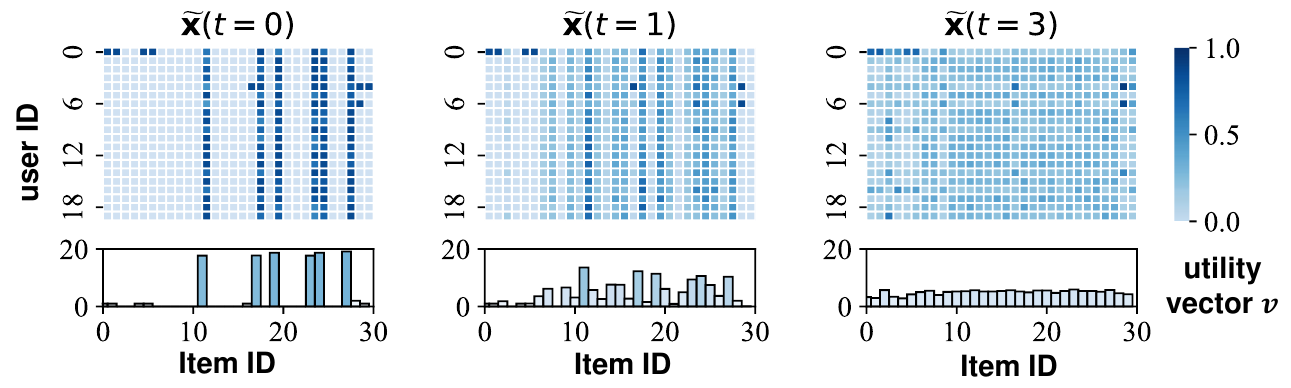}
    \caption{Visualization of \textit{Tax-rank} results under different tax rate $t=0,1,3$. The utility vector $\bm{x}$ is computed as the amortized value along the columns of the recommendation matrix.  }
    \label{fig:analysis_2}  
\end{figure}

\begin{figure}\
    \centering
     \subfigure[Price of taxation]
    {
        \includegraphics[width=0.46\linewidth]{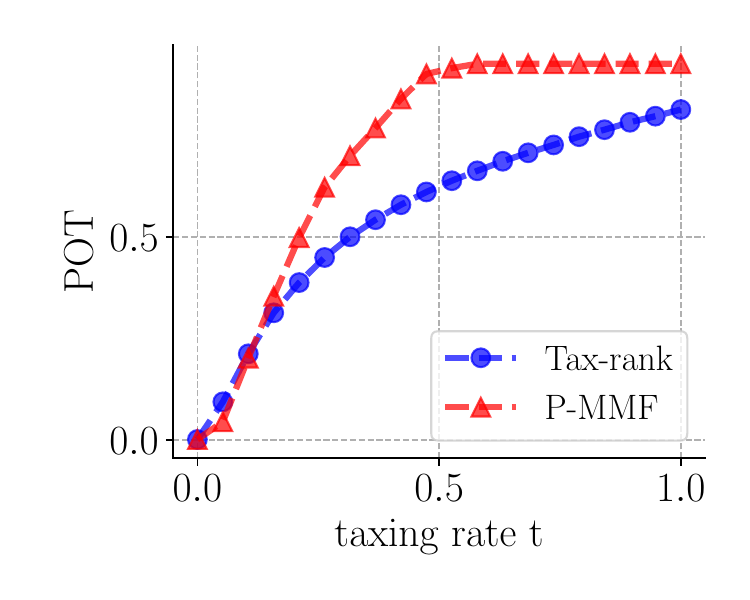}
    }
     \subfigure[Inference time comparison]
    {
        \includegraphics[width=0.46\linewidth]{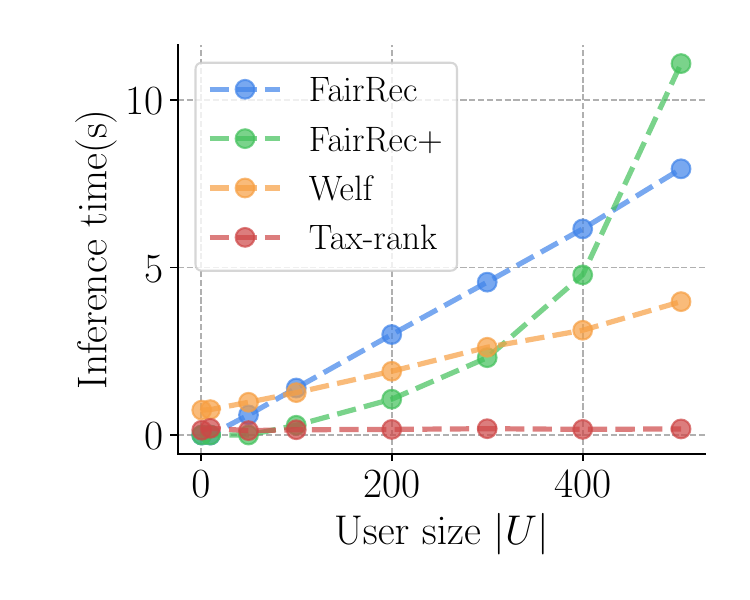}
    }
    \caption{Sub-figure (a) illustrates the price of taxation (POT) change w.r.t tax rate $t$. Sub-figure (b) describes online inference items for Tax-rank and other baselines w.r.t user size $|\mathcal{U}|$.}
    \label{fig:analysis_1}
\end{figure}


\subsection{Experimental Analysis}
We also conducted experiments to analyze our proposed Tax-rank on Yelp dataset for top-10 ranking. Similar phenomena can be observed in other datasets and with different top-k settings. 

\subsubsection{Continuity} Firstly, we experiment to assess the continuity of the proposed Tax-rank and other baseline methods by adjusting the tax rate $t$ or other parameters related to fairness. As shown in Figure~\ref{fig:Lorenz_Curve}, we draw the Lorenz Curve~\cite{Lorenzcurve} of three best-performing baselines and our model Tax-rank under different tax rate $t$ or other parameters. Since Lorenz Curve is a graphical representation that vividly illustrates income or wealth distribution within a population, providing insights into the level of inequality in economics.

Figure~\ref{fig:Lorenz_Curve} shows the proportion of overall exposure percentage assumed by the bottom $x$\% items. In simpler terms, the Lorenz Curve reveals the percentage (y\%) of the total exposures accumulated by the worst-off $x\%$ of items in the distribution. From Figure~\ref{fig:Lorenz_Curve}, it is evident that as the tax rate changes from $0$ to $\infty$, Tax-rank exhibits superior continuity. This is observed by the nearly averaged increase in the colored area between the two Lorenz Curves (can also be viewed as the fairness increasing value when increasing tax rate $t$). It indicates a more consistent response to a tenfold increase in the tax rate. On the contrary, as the tax rate changes, baselines, exhibit either minimal alterations (P-MMF) or an uneven increase (Welf and FairRec) in the colored area. 
This implies that a minor adjustment in the fair-aware parameter will result in a more significant performance alteration compared to Tax-rank.


\subsubsection{Controllability over accuracy loss.} Then, we experiment Figure~\ref{fig:analysis_1} (a) to demonstrate how the price fairness of the taxation (POT) changes of Tax-rank and baseline P-MMF \wrt variations in the tax rate $t$, ranging from $0$ to $1$. Since P-MMF is a represented and best-performing tax policy in Theorem~\ref{theo:previous_tax}. 

From the curve, it is evident that as we increase the tax rate $t$, our proposed Tax-rank and P-MMF approach leads to a reduction in the total ranking accuracy. 
Firstly, it is evident that the POT achieved by Tax-rank is notably lower compared to that of P-MMF, highlighting the effectiveness of the Tax-rank approach. Secondly, we can also observe that the POT of Tax-rank obeys a smooth form, as the theoretical analysis results in Theorem~\ref{theo:POF}. However, the previous ranking tax policies exhibit a more complex POT function form, resulting in poor systematic controllability.

\subsubsection{Inference time.} 
We conduct experiments to investigate the total inference time of the Tax-rank method compared to other item fairness baselines. In our analysis, our objective is to assess the total inference time across various user sizes $|\mathcal{U}|$, within real-world ranking applications. Therefore, we conducted tests to measure the total inference time of various models in terms of the varying number of users, all while keeping the number of items constant.

Figure~\ref{fig:analysis_1} (b) reports the curves of total inference time (s) w.r.t. user size $|\mathcal{U}|$. It is worth noting that the Tax-rank method demonstrates remarkably low inference times, typically taking less than ten million seconds across different user sizes. Furthermore, when compared to other baseline methods, the inference time of these alternatives tends to increase either linearly or exponentially with changing user sizes, whereas Tax-rank consistently maintains a low inference time.

This is attributed to the fact that the CP solver process (Equation~(\ref{eq:upperbound})) along with Tax-rank reduces the variable size from $|\mathcal{U}||\mathcal{I}|$ to $|\mathcal{I}|$ and employs the OT algorithm to efficiently map the solution to each user.
Therefore, Tax-rank method involves matrix operations with limited sensitivity to changes in user size.

\subsubsection{Visualizing fair re-ranking results.} 
In Figure~\ref{fig:analysis_2}, we visualize the Tax-rank recommendation results under the tax rate of $0, 1, 3$. In Figure~\ref{fig:analysis_2}, we visualize the ranking result matrix $\widetilde{\bm{x}}$, where the elements $\widetilde{\bm{x}}_{u,t}$ denotes the probability of recommending item $i$ to user $u$ in Equation~(\ref{eq:sinkhorn}). We also visualize the utility vector $\bm{v}$, which is the amortized value for the columns value of $\widetilde{\bm{x}}$ (\ie $\bm{v}=\sum_{u\in\mathcal{U}}\widetilde{\bm{x}}_{u,i}$). The utility vector $\bm{v}_i$ represents the utility value associated with the ownership of item $i$.

The results clearly demonstrate that the accuracy-first solution consistently ranks the most popular items highly for users, thereby enhancing overall utility but potentially leading to market dominance by 
a few top items. Regarding $t=1$, Tax-rank tends to distribute rankings to items in proportion to their contribution to the market. For $t=3$, Tax-rank method strives to provide equal exposure and similar utilities to every item in the ranking. The experiment also served as validation that Tax-rank method can effectively adapt to various fairness principles as intended.

\subsection{Ablation Experiments}\label{app:ablation}
In this section, we aim to conduct ablation experiments for Tax-rank. To better investigate the performance of our model under different parameter settings, we also conduct a series of ablation experiments on the Yelp dataset under ranking size $K=10$. Similar experiment results are also observed on other datasets and other ranking sizes $K$.

\begin{figure}
    \centering  
    \includegraphics[width=0.9\linewidth]{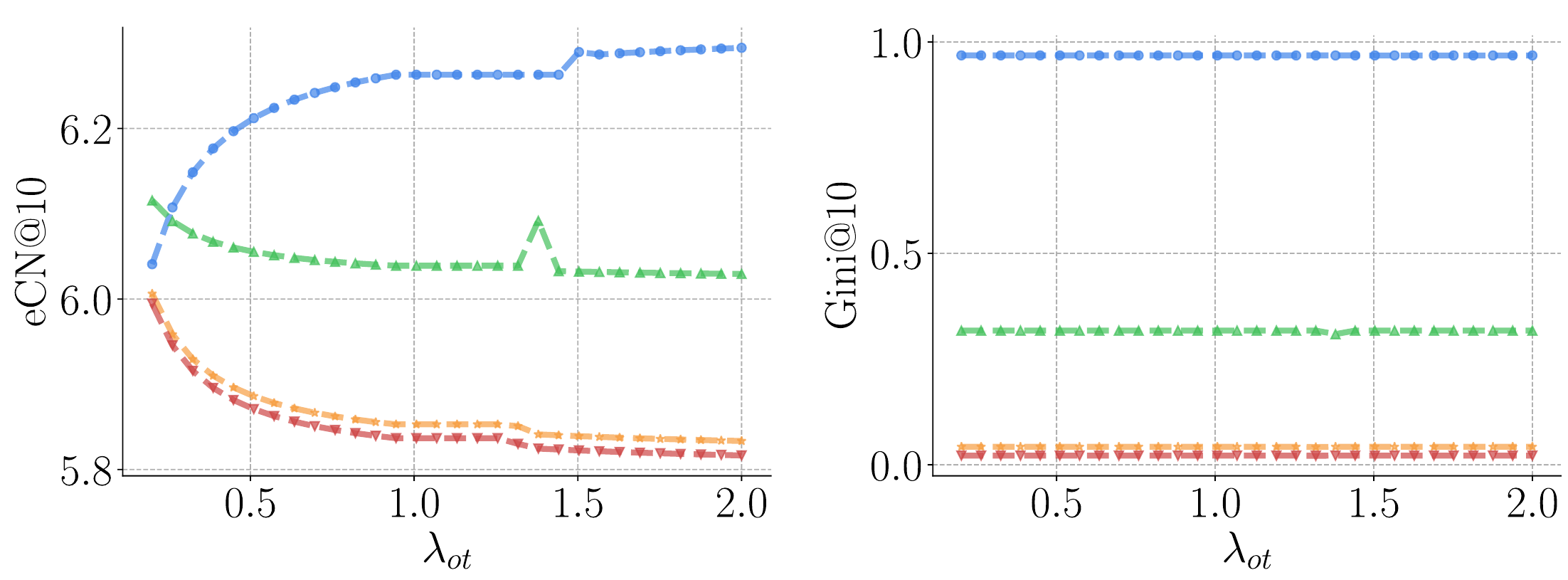}
    \caption{eCN and Gini change \wrt the coefficient of entropy regularizer $\lambda_{ot}$ in OT under different tax rate $t$.}
    \label{fig:abalation_1}
\end{figure}


\subsubsection{OT regularizer coefficient $\lambda_{ot}$.}

In this section, we first conduct the ablation study for the OT regularizer coefficient $\lambda_{ot}$ (Equation~(\ref{eq:sinkhorn}), which trade-off the entropy loss to force the variable $\bm{x}_{u,i}$ into the feasible region. Figure ~\ref{fig:abalation_1} illustrates how the ranking accuracy (eCN) and fairness metric (Gini) \wrt $\lambda_{ot}$ from $0.1-2.0$.

Specifically, from Figure ~\ref{fig:abalation_1} (a), we can observe that when the tax rate $t$ is small ($t<1$), the accuracy increases as $\lambda_{ot}$ becomes larger. When the tax rate $t$ is large ($t\ge1$), the accuracy will decrease as $\lambda_{ot}$ becomes larger. 
From Figure ~\ref{fig:abalation_1} (b), We can observe that Gini remains relatively stable across the different $\lambda_{ot}$. Moreover, as the tax rate $t$ increases, Tax-rank begins to prioritize fairness among items more prominently, resulting in a decrease in Gini.

The reason is the trade-off between the smoothness and convexity of the exposure probability distribution among items by tuning the parameter $\lambda_{ot}$. A higher $\lambda_{ot}$ will emphasize more on the smoothness of the probability distribution, implying a reduced disparity in exposure among items. When tax rate $t$ is relatively small ($t < 1$), a slight reduction in the difference in exposure probability among items allows some items with lower $w_{u,i}$ to gain more exposure, resulting in an increase in eCN. However, when $t$ becomes larger, the increased smoothness of probability distributions results in a decrease in eCN consequently.





\section{Conclusion}

In this paper, we re-conceptualize the fair re-ranking task by framing it as a taxation process from an economic perspective. In theory, we reformulate previous fair re-ranking architectures as a tax process, which imposes an additional tax on each item. However, we find that previous fair re-ranking methods fail to satisfy two crucial taxation attributes: continuity and controllability over accuracy loss. To address this challenge, we introduce a novel fair re-ranking model named Tax-rank. The optimization objective of Tax-rank levies taxes based on the difference in utility between two items. For optimization, we introduce a highly effective and efficient algorithm to optimize the Tax-rank objective, employing the Sinkhorn algorithm in optimal transport. Theoretical evidence shows that Tax-rank has superior continuity and systematic controllability compared to existing methods. Through extensive experiments on two publicly available datasets, it is evident that Tax-rank outperforms state-of-the-art baselines in terms of both fair ranking performance and efficiency.

\begin{acks}
This work was funded by the National Key R\&D Program of China (2023YFA1008704), the National Natural Science Foundation of China (No. 62377044, 62276248), the Youth Innovation Promotion Association CAS (No. 2023111), Beijing Key Laboratory of Big Data Management and Analysis Methods, Major Innovation \& Planning Interdisciplinary Platform for the  ``Double-First Class” Initiative, Public Computing Cloud, funds for building world-class universities (disciplines) of Renmin University of China. Supported by the Outstanding Innovative Talents Cultivation Funded Programs 2024 of Renmin University of China.
\end{acks}

\newpage
\appendix

\section*{Appendix}
\section{Proof of Theorem~\ref{theo:previous_tax}}\label{app:previous_tax}
\begin{proof}
    Let $\bm{a}_{[i]}$ denotes the $i$-th largest element of $\bm{a}$. Firstly, we will prove Equation~(\ref{eq:constraint}) can be written as the same form of Equation~(\ref{eq:lambda}).
    By utilizing the Lagrange multiplier method, we can write the objective of Equation~(\ref{eq:constraint}) as
    \[
    W_2(\bm{x}) = \min_{\lambda_2 \ge 0}\max_{\bm{x}\in \mathcal{X}} \sum_{i\in\mathcal{I}} \gamma_i\sum_{u\in\mathcal{U}} w_{u,i}\bm{x}_{u,i} + \lambda_2 (r(\bm{v})-\pi),
    \]
    where optimal value is $\argmin_{\lambda_2} \left[\sum_{u\in\mathcal{U}}\sum_{k=1}^K c_{n,[k]} +  r(\bm{v})\lambda_2-\lambda_2\pi \right]$, $c_{n,i} = \gamma_i w_{u,i}$. Therefore, we can only analyze the form of Equation~(\ref{eq:lambda}). Then following~\citet{xu2023p, balseiro2021regularized}, we can also utilize the Lagrangian condition to break the relationship between $\bm{v}$ and $\sum_{u\in\mathcal{U}}\sum_{i\in\mathcal{I}} s_{t,i}x_{t,i}$:
    $
         W_1(\bm{x}) = \max_{\bm{x}\in \mathcal{X}} \sum_{u\in\mathcal{U}}\sum_{i\in\mathcal{I}} s_{t,i}\bm{x}_{u,i} + g(\mu^*),
    $
    where $s_{t,i} = \gamma_i w_{u,i} + \bm{\mu}_i^*, g(\mu)=\max_{v} \left(\lambda r(\bm{v}) + \sum_{i\in\mathcal{I}}\bm{\mu}_i\bm{v}_i\right)$ and 
    \[
    \mu^*=\argmin_{\bm{\mu}}\max_{\bm{x}\in \mathcal{X}} \sum_{u\in\mathcal{U}}\sum_{i\in\mathcal{I}} (\gamma_i w_{u,i} + \bm{\mu}_i)\bm{x}_{u,i} + g(\mu).
    \]
   
\end{proof}


\section{Proof of Theorem~\ref{theo:non-continuty}}\label{app:non-continuty}
\begin{proof}
    According the Theorem~\ref{theo:previous_tax}, we can know the Equation~(\ref{eq:constraint}) and Equation~(\ref{eq:lambda}) is equivalent. Therefore, we mainly analyze the continuity of regularizer-based form in Equation~(\ref{eq:lambda}):
    $W_1(\bm{x};\lambda) = \max_{\bm{x}\in \mathcal{X}} \sum_{u\in\mathcal{U}}\sum_{i\in\mathcal{I}} n_{u,i}\bm{x}_{u,i} + r(\bm{v})\lambda$. Then we will prove $W_1$ is not continuous \wrt tax rate $\lambda$. 
    
    We can observe that the binary ranking solution $\bm{x}_{u,i}$ also leads to non-continuity. We can see that if the function $W_1(\bm{x};\lambda)$ changes, it implies that there exists an item $i_r$ being removed from a $u$'s ranking list, while simultaneously a new item $i_k$ is introduced and added to that $u$'s ranking list. Then the tax rate $\lambda$'s changing value $\delta \lambda$ rate should be at least 
    \[
        \delta \lambda > \frac{|n_{u,i_r}-n_{u,i_k}|}{|\delta r(\bm{v})|}
    \]
    to result a different recommendation results, where $\delta r(\bm{v})$ is the change value of fairness value when item $i_r$ is replaced with item $i_k$. Since the $r(\bm{v})$ is also often non-continuous, such as max-min form~\cite{xu2023p}, therefore, we have $\exists M>0, \delta r(\bm{v})>M$.
    

\end{proof}


\section{Proof of Theorem~\ref{theo:POF}}\label{app:POF}

    \begin{lemma}\label{lemma:1}
    Given the vector $\bm{x}\in\mathbb{R}^N$, $\sum_{i=1}^N w_i\bm{x}_i^{1-t} \ge C$, for any $t>0, w_i>0, \bm{x}_i>0, \sigma=\min_i \bm{x_i}$ we have 
    \begin{equation}
        \sum_{i=1}^N w_i\bm{x}_i \ge C\sigma^{t} N^{-\frac{t}{1+t}},
    \end{equation}
\end{lemma}

\begin{proof}

    We apply Hölder inequality, we have when $\frac{1}{1+t} + \frac{t}{1+t} = 1$,
    \begin{equation}\label{eq:holder}
        C \leq \sum_{i=1}^N w_i\bm{x}_i*\frac{1}{\bm{x}_i^{t}} \leq (\sum_{i=1}^N w_i^{1+t}\bm{x}_i^{1+t})^{\frac{1}{1+t}}(\sum_{i=1}^N \frac{1}{\bm{x}_i^{1+t}})^{\frac{t}{1+t}}.
    \end{equation}

    Also, we have
    $
        (\sum_{i=1}^N w_i\bm{x}_i)^{1+t} \ge \sum_{i=1}^N w_i^{1+t} \bm{x}_i^{1+t},
    $
    therefore, we have 
    $
        (\sum_{i=1}^N w_i^{1+t}\bm{x}_i^{1+t})^{\frac{1}{1+t}} \leq \sum_{i=1}^N w_i\bm{x}_i
    $
    Let $\sigma = \min_i \bm{x}_i$, we have
    \begin{equation}\label{eq:reduce2}
        (\sum_{i=1}^N \frac{1}{\bm{x}_i^{1+t}})^{\frac{t}{1+t}} \leq \frac{N^{\frac{t}{1+t}}}{\sigma^t}.
    \end{equation}

    Then combining Equation~(\ref{eq:reduce2}) and Equation~(\ref{eq:holder}) becomes:
    \[
        \sum_{i=1}^N w_ix_i  \ge C\sigma^{t} N^{-\frac{t}{1+t}}.
    \]
    
\end{proof}

Then we will give a formal proof of Theorem~\ref{theo:POF}.
\begin{proof}
    From Lemma~\ref{theo:upperbound}, we can also write  :
      \begin{equation}
        \begin{aligned}
            f(\bm{x}; t) \ge \widetilde{f}(\bm{x}; t) &\max_{\bm{e}\in\mathcal{E}} \sum_{i} m_i\frac{\bm{e}_i^{1-t}}{1-t},\\
        \end{aligned}
    \end{equation}
    where $\mathcal{E} = \{\bm{e}|\sum_i\bm{e}_i=K,\quad  0 \leq \bm{e}_i \leq 1\}$, the input $\bm{e}$ optimal value is $\bm{z}$, which represents be the best allocation under the $t$-fairness criterion. We briefly let  $\widetilde{\text{Acc}}(t)$ denotes the accuracy value of $\widetilde{f}(\bm{x};t)$.

    Firstly, we will bound the $\widetilde{\text{Acc}}(t)$:
    without generality, we assume:
    \begin{equation}\label{eq:increasing_order}
        m_1\bm{z}_1 \ge  m_2\bm{z}_2 \ge \cdots, \ge m_{|\mathcal{I}|}\bm{z}_{|\mathcal{I}|}.
    \end{equation}
    The necessary first-order condition for the optimality of $\bm{e}$ can be expressed as:
    \[
        \nabla f(\bm{z}; t)(\bm{e}-\bm{z}) \leq 0, \forall \bm{e}\in \mathcal{E},
    \]
    
    The equation can be equivalently written as:
    \begin{equation}\label{eq:first-condition}
        \bm{g}^{\top}\bm{e} \leq 1, \forall \bm{e}\in \mathcal{E},
    \end{equation}
    where
    $
        \bm{g}_i = \frac{m_i\bm{z}_i^{-t}}{\sum_{i}m_i\bm{z}_i^{1-t}}.
    $
    We observe the Equation~(\ref{eq:first-condition}), which is a well-studied knapsack problem~\cite{salkin1975knapsack}, with the best solution:
    \begin{equation}\label{eq:con1}
       \frac{ \sum_{k=1}^K m_{|{\mathcal{I}}|-k+1}\bm{z}_{|{\mathcal{I}}|-k+1}^{-t}}{\sum_{i}m_i\bm{z}_i^{1-t}} \leq 1,
    \end{equation}
    since according the Equation~(\ref{eq:increasing_order}), we have
    \[ 
        m_1\bm{z}_1^{-t} \leq\  m_2\bm{z}_2^{-t} \leq \cdots, \leq m_{|\mathcal{I}|}\bm{z}_{|\mathcal{I}|}^{-t}.
    \]

    From the  Equation~(\ref{eq:con1}) , exists $0<\lambda<1$, we have:
    \begin{equation}\label{eq:con2}
        \sum_i m_i\bm{z}_i^{1-t} \ge \lambda \text{Acc}(0) \bm{z}_1^{-t},
    \end{equation}
    \[
        \sum_{k=1}^K m_{|{\mathcal{I}}|-k+1}\bm{z}_{|{\mathcal{I}}|-k+1}^{-t} \ge \bm{z}_1^{-t} \lambda \sum_{k=1}^K m_{[k]} = \lambda \text{Acc}(0) \bm{z}_1^{-t},
    \]
    where $m_{[k]}$ denotes the $k$-th largest element of $m_i$.

    Taking the Equation~(\ref{eq:con2}) into Lemma~\ref{lemma:1}, we have
    \begin{equation}\label{eq:con3}
        \widetilde{\text{Acc}}(t)=\sum_i^{|\mathcal{I}|} m_i\bm{z}_i \ge \lambda \text{Acc}(0) (\frac{\bm{z}_{|\mathcal{I}|}}{\bm{z}_1})^{t} N^{-\frac{t}{1+t}}.
    \end{equation}
    Therefore,
    \begin{align*}
        \text{POT} &= \frac{\text{Acc}(0)-\text{Acc}(t)}{\text{Acc}(0)} \leq \frac{\text{Acc}(0)-\widetilde{\text{Acc}}(t)}{\text{Acc}(0)}\\
        &\leq  1- \lambda(\frac{\bm{z}_{|\mathcal{I}|}}{\bm{z}_1})^{t} N^{-\frac{t}{1+t}}\\
       & = 1 - O(N^{-\frac{t}{1+t}}).
    \end{align*}
    
\end{proof}

\bibliographystyle{ACM-Reference-Format}
\newpage
\bibliography{ref}

\end{document}

%% file: math_commands.tex
\usepackage{amsmath,amsfonts,bm}

\def\eqref#1{equation~\ref{#1}}

\def\1{\bm{1}}

\DeclareMathAlphabet{\mathsfit}{\encodingdefault}{\sfdefault}{m}{sl}
\SetMathAlphabet{\mathsfit}{bold}{\encodingdefault}{\sfdefault}{bx}{n}

\DeclareMathOperator*{\argmax}{arg\,max}
\DeclareMathOperator*{\argmin}{arg\,min}